\documentclass[a4paper,english]{lipics-v2018}

\newif\ifArxivVersion
\ArxivVersiontrue

\ifArxivVersion
    \usepackage[appendix=append]{apxproof}
    \hideLIPIcs    
    \relatedversion{This is the full version with appendix of the conference
        paper with the same name presented at SOCG 2018.}
    \ArticleNo{M}
    \EventShortTitle{arxiv version}
\else
    \usepackage[appendix=strip,bibliography=common]{apxproof}
    \relatedversion{The full version with appendix may be obtained from\\
        \texttt{http://arxiv.org/abs/TODO}}
    \EventEditors{Bettina Speckmann and Csaba D. T{\'o}th}
    \EventNoEds{2}
    \EventLongTitle{34th International Symposium on Computational Geometry (SoCG 2018)}
    \EventShortTitle{SoCG 2018}
    \EventAcronym{SoCG}
    \EventYear{2018}
    \EventDate{June 11--14, 2018}
    \EventLocation{Budapest, Hungary}
    \EventLogo{socg-logo}
    \SeriesVolume{99}
    \ArticleNo{60}
\fi

\usepackage[utf8]{inputenc}
\usepackage[T1]{fontenc}
\usepackage{commath,microtype}

\title{Random Walks on Polytopes of Constant Corank}
\author{Malte Milatz}{ETH Zürich\\Switzerland}{mmilatz@inf.ethz.ch}{}{}
\authorrunning{M.~Milatz}
\Copyright{Malte Milatz}
\subjclass{
    \ccsdesc[500]{Theory of computation~Randomness, geometry and discrete structures},
    \ccsdesc[300]{Theory of computation~Linear programming}
}
\keywords{polytope, unique sink orientation, grid, random walk}

\theoremstyle{plain}
\newtheoremrep{lemma}[theorem]{Lemma}
\newtheoremrep{corollary}[theorem]{Corollary}
\allowdisplaybreaks
\nolinenumbers

    \newcommand{\NN}{\mathbb{N}}
    \newcommand{\RR}{\mathbb{R}}
    \newcommand{\Cc}{\mathcal{C}}
    \newcommand{\Ee}{\mathcal{E}}
    \newcommand{\Ll}{\mathcal{L}}
    \newcommand{\Rr}{\mathcal{R}}
    \newcommand{\one}{\mathbf{1}}
    \newcommand{\zero}{\mathbf{0}}
    \newcommand{\ee}{\mathbf{e}}
    \newcommand{\T}{\mathrm{T}}
    \DeclareMathOperator{\E}{E}
    \newcommand{\aff}{\mathrm{aff}}
    \newcommand{\phase}{\mathrm{phase}}
    \newcommand{\conv}{\operatorname{conv}}
    \makeatletter
        \newcommand*{\defeq}
            {\mathrel{\rlap{\raisebox{0.3ex}
            {$\m@th\cdot$}}\raisebox{-0.3ex}{$\m@th\cdot$}}=}
    \makeatother
    \newcommand{\ex}[1]{\E\bigl[#1\bigr]}
    \newcommand{\pr}[1]{\Pr\bigl[#1\bigr]}
    \newcommand{\Given}{\,\big|\,}
    \newcommand{\pmat}[1]{\begin{pmatrix}#1\end{pmatrix}}
    \newcommand{\chooseTwo}[1]{{\textstyle\binom{#1}{2}}}
    \newcommand{\sentenceDash}{\:\textemdash\:}
    \newcommand{\midSentenceDash}{\,\textemdash\,}

\begin{document}

\maketitle

\begin{abstract}
    We show that the pivoting process associated with one line and $n$ points
    in $r$-dimensional space may need $\Omega(\log^r n)$ steps in expectation
    as $n \to \infty$.
    The only cases for which the bound was known previously were for $r \le 3$.
    Our lower bound is also valid for the expected number of pivoting steps in
    the following applications:
    (1)~The \textsc{Random-Edge} simplex algorithm on linear programs
    with $n$ constraints in $d = n-r$ variables; and
    (2)~the directed random walk on a grid polytope of corank $r$
    with $n$ facets.
\end{abstract}

\begin{figure}[b]
    \ifArxivVersion
        \vspace{17pt} 
    \fi
    \begin{minipage}[t]{0.2\textwidth}
        \vspace{0pt}
        \begin{subfigure}{\textwidth}
            \centering
            \includegraphics[height=6em]{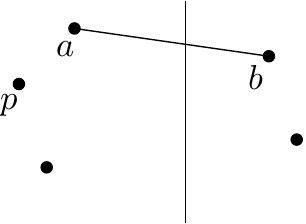}
            \vspace{0pt}
            \caption{} 
            \label{Figure: one line and n points}
        \end{subfigure}
    \end{minipage}
    \hspace{0.1\textwidth}
    \begin{minipage}[t]{0.25\textwidth}
        \vspace{0pt}
        \begin{subfigure}{\textwidth}
            \centering
            \includegraphics[height=6em]{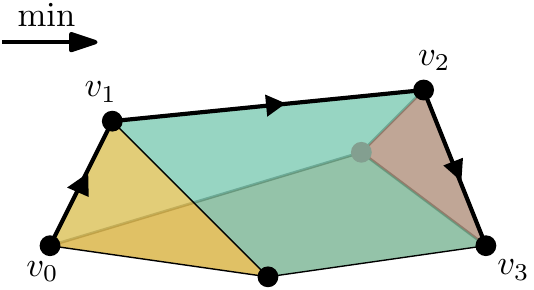}
            \vspace{0pt}
            \caption{}
            \label{Figure: polytope}
        \end{subfigure}
    \end{minipage}
    \hfill
    \begin{minipage}[t]{0.3\textwidth}
        \vspace{0pt} 
        \caption{\textbf{(a)} An instance of ``one line and $n$ points''.
            \textbf{(b)} A possible path that the directed random walk might
            take on the corresponding polytope.}
    \end{minipage}
\end{figure}

\section{Introduction}
\label{Section: introduction}

In 2001 Gärtner et al.~\cite{GaertnerSTVW'03} introduced a random process
involving a vertical line and $n$ points in the plane (see Figure~\ref{Figure:
one line and n points}) which they called ``the fast process''
or just ``one line and $n$ points''.
The starting position of this process is a pair $\{a,b\}$ of points that lie
on opposite sides of the vertical line.
In each step the process picks a point $p$ (a pivot) uniformly at random from
the points that lie below the non-vertical line $ab$.
The subsequent position is then the unique pair $\{p,q\}$ with $q \in \{a,b\}$
such that $p$ and $q$ lie again on opposite sides of the vertical line.
For example, in Figure~\ref{Figure: one line and n points} the next position
would be $\{p,b\}$.
The authors of \cite{GaertnerSTVW'03} gave matching upper and lower bounds of
order $\log^2(n)$ for the expected duration of this process in the plane.

The process generalizes naturally to higher dimensions.
However, understanding its behavior in any dimension other than $2$ has remained a wide
open problem, with the notable exception of an $\Omega(\log^3 n)$ lower bound
in three dimensions \cite{Tschirschnitz'03}.
As the dimension grows, the situation indeed becomes increasingly complicated,
and the question has seen no further improvements for the subsequent fifteen
years.

The relevance of the generalized process lies in its intimate connection with
the \textsc{Random-Edge} simplex algorithm for linear programming,
which has already been widely considered before;
for example cf.~\cite{BroderDFRU'95,FelsnerGT'05,GaertnerHZ'98,Hoke'88,MatousekSz'06}.
\textsc{Random-Edge} is naturally formulated in terms of a
random walk on the vertex-edge graph of the feasible region
(see Figure~\ref{Figure: polytope}).
To be precise it is a \emph{directed} random walk, because every edge may be
used in one prescribed direction only: the direction along which the objective
function improves.
In each step the directed random walk moves from the current vertex to a vertex
chosen uniformly at random from all neighbors with smaller objective value.
By means of the so-called \emph{extended Gale transform}, the directed
random walk on a $d$-polytope with $n$ facets translates precisely into the
process of ``one line and $n$ points'' in $\RR^{r}$, where $r \defeq n-d$
denotes the \emph{corank} of the polytope.
For lack of space we refer the interested reader to the appendix of
\cite{GaertnerSTVW'03} for a complete exposition of the extended Gale
transform.

The interpretation in terms of polytopes also explains the difficulties that
arise when analyzing the process for $r \ge 3$.
Namely, it is known that every simple polytope of corank $r=1$ is a simplex,
and every simple polytope of corank $r=2$ is a product of two simplices; these
situations are thus well-understood, classified, and not too complicated.
Already for $r=3$, however, the classification is considerably more involved
(cf.~chapter 6.3 in \cite{Gruenbaum'03}), and for $r \ge 4$ no similar
classification exists.

The name ``corank'' is not entirely standard; it has been used e.g.~in
\cite{PfeifleZ'04}.
Despite its anonymity it plays a prominent role in polytope theory: The
\emph{Hirsch conjecture} once stated that the corank might be an upper bound
on the diameter of any polytope.
Since the Hirsch conjecture, in its strict form, has been disproved by Santos
\cite{Santos'12}, the search for a close connection between these quantities
continues.
Indeed we believe it fruitful to analyze algorithms for linear
programming in a setting where the corank is assumed to be small, i.e., a
constant or a slowly growing function of $n$.
A positive result of this type is for example Kalai's \textsc{Random-Facet}
algorithm:
His upper bound in \cite{Kalai'92} for the expected number of arithmetic
operations becomes polynomial if the corank is taken to be of order
$r = O(\log n)$.
In contrast, \textsc{Random-Edge} has proved to be notoriously hard to
analyze, and tight bounds are rare when we want to understand the behavior of
a given simplex algorithm on a complete class of instances.
The analysis in \cite{KaibelMSZ'05} for $d=3$ suggests that
\textsc{Random-Edge} can be a good choice in low dimension;
the present paper takes the dual viewpoint, fixing the corank rather than the
dimension.
Note that the mildly exponential lower bound obtained by Friedmann et
al.~\cite{FriedmannHZ'11} does not pose any restrictions on the corank.

A polytope with $n$ facets and constant corank $r$ has $O(n^r)$ vertices, and
this bound is tight; this follows for example from McMullen's theorem \cite{McMullen'71}.
Thus, $O(n^r)$ is a trivial bound for the number of \textsc{Random-Edge}
pivoting steps.
The known bounds for $r=2$ suggest that the process outperforms the trivial
bound considerably; however, it is not at all clear what can be expected in
general.
It is conceivable that a bound of $O(\log^r n)$ holds, although we
are currently missing the mathematical techniques and insights to prove this.
In this paper, we prove that one can at least not do better.

The history of lower-bound constructions for specific LP-solving algorithms
shows that it is often considerably easier to prove bounds in an
abstract model; this is how \emph{unique sink orientations} (USOs) enter the picture.
The same principle applies to the present paper:
We first present a construction in the USO model in Section~\ref{Section:
prelude} and strengthen the result to the geometric setting in the rest of the
paper.
Somewhat atypically, the main ideas underlying the geometric construction
are not \emph{entirely} different from the ideas underlying the USO construction.

\subparagraph{Our results.} We prove that the random process of
``one line and $n$ points'' in $\RR^r$ may take $\Omega(\log^r n)$ steps in
expectation.
This generalizes previous constructions for the cases $r=2,3$
\cite{GaertnerSTVW'03, Tschirschnitz'03}.
The exact bound that we obtain is specified in Theorem~\ref{Theorem: geometric
lower bound}.
Using the extended Gale transform, our result can be rephrased in different
settings, as in the following theorems.
Theorem~\ref{Theorem: linear programs} rephrases our result in the language of
linear programs.
Theorem~\ref{Theorem: polytopes} uses instead the language of polytopes.
The combinatorial type of the polytopes that we construct in this way is
rather special: they are \emph{grid polytopes}, i.e., Cartesian products of
simplices.

\begin{theorem}
    \label{Theorem: linear programs}
    Let $r \in \NN_0$ be a fixed parameter.
    There are linear programs in $d$ variables with $n = d+r$ constraints
    on which the \textsc{Random-Edge} algorithm needs $\Omega(\log^r n)$
    pivoting steps in expectation as $n \to \infty$.
\end{theorem}

\begin{theorem}
    \label{Theorem: polytopes}
    Let $r \in \NN_0$.
    There are grid $d$-polytopes with $n = d+r$ facets on which
    the directed random walk (with the direction specified by a linear
    function) needs $\Omega(\log^r n)$ steps in
    expectation as $n \to \infty$.
\end{theorem}

\section{Prelude: Walks on grids}
\label{Section: prelude}

In this section we prove a lower bound of order $\log^r(n)$ for the expected
duration of a random walk on a certain class of directed graphs, namely
\emph{unique sink orientations of grids}.
The construction does not involve any geometry and should be simpler to
understand than the point-set construction in Section~\ref{Section:
construction}.

Given a directed (multi-)graph $G$ and a vertex $v_0$ of $G$, the
\emph{directed random walk} is the random process $v_0,v_1,v_2,\dots$
described as follows:
If the current position is $v_i$, choose one of the outgoing edges at
$v_i$ uniformly at random, and let $v_{i+1}$ be the other endpoint of that
edge.
The process terminates when (and if) it reaches a sink.
The random variable
\begin{align*}
    T(G,v_0) =
    \min \{ t \,:\, v_t \text{ is a sink} \}
\end{align*}
will denote the duration of the directed
random walk on $G$ starting in $v_0$.
We will abbreviate $T(G) := T(G,v_0)$ for a starting position $v_0$ chosen
uniformly at random from the set of vertices.

Following the terminology in \cite{GaertnerMR'08}, a \emph{grid} is a
Cartesian product of a (finite) number of (finite) complete graphs
$G_1,\dots,G_r$.
Explicitly, the vertex set of this graph is $V(G_1) \times \dots \times
V(G_r)$, and two vertices $(u_1,\dots,u_r)$ and $(v_1,\dots,v_r)$ are joined
by an edge if and only if they differ in exactly one coordinate.
We could have defined grids equivalently as the vertex-edge graphs of grid
polytopes.
The number $r$ is usually called the \emph{dimension} of the grid (this is
\emph{not} the dimension of the underlying grid polytope!); and its
\emph{size} is the number $|V(G_1)| + \dots + |V(G_r)|$.

A \emph{subgrid} of a grid is an induced subgraph on a set of the form
$U_1 \times \dots \times U_r$ with $U_i \subseteq V(G_i)$.
Finally, a \emph{unique sink orientation} of a grid is an orientation of the
edges of the grid with the property that every subgrid possesses a unique
sink.
An example is shown in Figure~\ref{Figure: grid USO}.

\begin{theorem}
    \label{Theorem: USOs}
    Let $r \in \NN_0$ and $m \in \NN$.
    There is an $r$-dimensional grid $G$ of size $n \defeq rm$, endowed with an
    acyclic unique sink orientation, such that a directed random walk on the
    grid, starting at a random position, takes at least
    \begin{equation}
        \label{eq: theorem USOs}
        \ex{ T(G) } \ge \frac{1}{r!} \ln^r\del{ m+1 } - 1
    \end{equation}
    steps in expectation.
\end{theorem}

We remark that the theorem is meaningful primarily for a fixed dimension $r$
and with the grid size tending to infinity.
In this setting it says that the number of pivoting steps can be of order
$\Omega(\log^r n)$.

    \begin{figure}
        \begin{minipage}[t]{0.6\textwidth}
            \begin{subfigure}{0.25\textwidth}
                \includegraphics[height=4em]{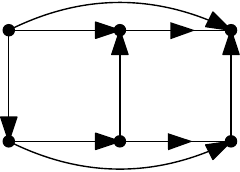}
                \vspace{4em}
                \caption{} 
                \label{Figure: grid USO}
            \end{subfigure}
            \hspace*{0.06\textwidth}
            \begin{subfigure}{0.25\textwidth}
                \centering
                \includegraphics[height=7em]{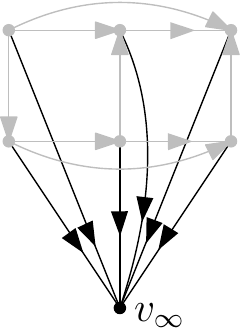}
                \vspace{1em}
                \caption{}
                \label{Figure: augmented}
            \end{subfigure}
            \hspace*{0.05\textwidth}
            \begin{subfigure}{0.25\textwidth}
                \centering
                \includegraphics[height=8em]{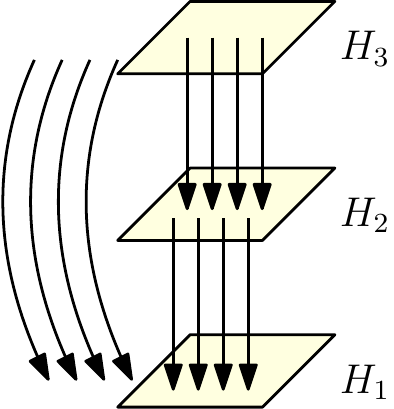}
                \caption{}
                \label{Figure: combed}
            \end{subfigure}
        \end{minipage}
        \hfill
        \begin{minipage}{0.34\textwidth}
                \caption{
                    \\\textbf{(a)}~A grid $G$ of dimension $r=2$ and size $n = 5$.
                    \\\textbf{(b)}~The augmented \mbox{(multi-)}graph $G^\Delta$, here for $\Delta=1$.
                    \\\textbf{(c)}~A schematic depiction of the grid constructed in the proof of
                        Theorem~\ref{Theorem: USOs}.
                }
        \end{minipage}
    \end{figure}

\begin{proof}
    We can choose the grid to be $G = G_1 \times \dots \times G_r$ with
    $
        G_1 = \dots = G_r = K_m.
    $
    We construct a unique sink orientation on $G$ by induction on $r$.
    For $r = 0$, the graph consists of a single vertex, so we need not orient
    any edges.

    For $r \ge 1$, we first choose a permutation of $V(G_r)$
    uniformly at random from the set of all permutations, and we label the
    vertices according to the chosen permutation,
    as in $V(G_r) = \{ 1, \dots, m \}$.
    Now we partition the grid into ``hyperplanes''
    \[
        H_i = G_1 \times \dots \times G_{r-1} \times \{ i \}
        \quad \textstyle (i = 1,\dots,m).
    \]
    Each hyperplane $H_i$ is a subgrid of dimension $r-1$, so we can
    inductively assign an orientation to it.
    We do this for each hyperplane independently (i.e., all random
    permutations used throughout the construction are chosen independently).
    The only edges that we still need to orient are those between vertices
    from two distinct hyperplanes.
    We orient those according to our chosen permutation of $V(G_r)$.
    Explicitly, the edge from a vertex $(u_1,\dots,u_{r-1},i)$ to
    $(u_1,\dots,u_{r-1},j)$ is directed forwards if and only if $i > j$.
    See Figure~\ref{Figure: combed} for an illustration.

    It is easy to verify that this defines an acyclic unique sink orientation
    of the grid.
    We will now analyze the duration of the random walk from a random starting
    position.
    Concerning the starting position, here is a key observation:
    Due to the random permutations involved in the construction of the grid
    orientation, it amounts to the same random process whether we start the
    walk in a random position of the grid, or whether we start in any fixed
    given position.
    Consequently, if the random walk visits one of the hyperplanes $H_i$, we
    can relate the behavior within $H_i$ to the $(r-1)$-dimensional
    construction, for which we have a lower bound by induction.
    However, there \emph{is} a difference between $H_i$ and the
    $(r-1)$-dimensional construction:
    namely, every vertex of $H_i$ has $i-1$ additional outgoing edges by which
    the walk may leave $H_i$ at any moment.
    To account for these, we make use of the following ``augmented''
    multigraph; see Figure~\ref{Figure: augmented} for an example.

    \begin{definition}
        Given any (multi-)graph $\Gamma$ and a parameter $\Delta \in \NN_0$,
        we define an \emph{augmented multigraph} $\Gamma^\Delta$ as follows.
        We add a new, special, vertex $v_\infty$ to the vertex
        set of $\Gamma$.
        Furthermore we add $\Delta$ many edges $\overrightarrow{v v_\infty}$
        for every $v \in V(\Gamma)$.
        If $\Delta = 0$ then we add one additional edge $\overrightarrow{s
        v_\infty}$ for every sink $s$ of $\Gamma$; this way we ensure that
        $v_\infty$ is the only sink of $\Gamma^\Delta$.
    \end{definition}

    \begin{lemma}
        Let $\Delta \in \NN_0$.
        Then the expected duration of the directed random walk on the
        augmented construction $G^\Delta$, starting from a random position
        (and ending in $v_\infty$), satisfies the bound
        \[
            \ex{ T(G^\Delta) } \ge \frac{1}{r!} \del{ \ln\del{ m + \Delta + 1 } - \ln( \Delta + 1) }^r.
        \]
    \end{lemma}

    \textit{Proof of the lemma.}
        We proceed by induction on $r \ge 0$.
        If $r=0$ then $G$ consists of a single vertex and there is nothing to
        prove; so let $r \ge 1$.
        For $i \in \{1,\dots,m\}$, let $T_i$ denote the number of positions
        that the walk visits on $H_i$, so that the total duration of the walk
        is given by
        \begin{align}
            \label{eq: sum of Ti}
            T(G^\Delta) = T_1 + \dots + T_{m}.
        \end{align}
        Furthermore, let $\Ee_i$ denote the event that the random walk visits
        at least one vertex in $H_i$ ($i = 1,\dots,m$).
        We claim
        \begin{align}
            \label{eq: claim}
            \pr{ \Ee_i } \ge \frac{1}{\Delta+i}.
        \end{align}
        To this end we consider the hitting time
        \[
            \tau_i := \min \cbr{ t ~:~ v_t \in \{ v_\infty \} \cup H_1 \cup \dots \cup H_{i} }
        \]
        where, as before, $v_t \in V(G^\Delta)$ denotes the position that the
        random walk visits at time $t$ ($t = 0,1,2,\dots$).
        Note that the hyperplanes $H_i$ are visited in decreasing order; so
        either the walk visits $H_i$ at time $\tau_i$, or not at all.
        Hence, $\Ee_i$ equals the event $\{ v_{\tau_i} \in H_i \}$.
        We now calculate the probability of this event by conditioning on $\tau_i \ge 1$.

        \textit{Case 1:} $\tau_i \ge 1$.
        Since $v_{\tau_i - 1}$ has $\Delta$ outgoing edges to $v_\infty$ and
        one outgoing edge to each of the hyperplanes $H_1,\dots,H_i$, and
        since the random walk is equally likely to move along any of these
        $\Delta+i$ edges, we obtain
        \begin{align}
            \pr{\Ee_i \Given \tau_i \ge 1} =
            \pr{ v_{\tau_i} \in H_i \Given \tau_i \ge 1 }
            = \frac{1}{\Delta + i}.
            \label{eq: claim-case1}
        \end{align}

        \textit{Case 2:} $\tau_i = 0$.
        Here we need to look at $v_0$, which (conditioned on $\tau_i = 0$) is
        a vertex taken uniformly at random from the set $H_1 \cup \dots \cup
        H_i$.
        Since the hyperplanes $H_1,\dots,H_i$ are all of equal cardinality,
        we obtain \begin{align} \pr{\Ee_i \Given \tau_i = 0} =
            \pr{v_0 \in H_i \Given \tau_i = 0 }
            = \frac{1}{i} \ge
            \frac{1}{\Delta + i}.
        \label{eq: claim-case2}
        \end{align}

        The claim \eqref{eq: claim} follows by combining \eqref{eq:
        claim-case1} and \eqref{eq: claim-case2}.
        Now, it is easy to see that $T_i | \Ee_i$ has the same distribution as
        $T((H_i)^{\Delta+i-1})$, so that we obtain
        \begin{align*}
            \ex{ T_i } &= \pr{\Ee_i} \cdot \ex{ T_i \Given \Ee_i }
            \\ &\ge \frac{1}{\Delta+i} \cdot \ex{ T( (H_i)^{\Delta+i-1} ) }
            \\      &\ge \frac{1}{\Delta+i} \cdot \frac{1}{ (r-1)! } \del{ \ln\del{ m + \Delta + i } - \ln( \Delta+i ) }^{r-1}
        \end{align*}
        where the last step was using the induction hypothesis.
        With \eqref{eq: sum of Ti} we obtain
        \begin{align*}
            \ex{T} &=    \sum_{i=1}^{m} \ex{T_i}
             \\ &\ge  \sum_{i=1}^{m} \frac{1}{ \Delta+i } \cdot \frac{1}{ (r-1)! } \del{ \ln\del{ m + \Delta + i } - \ln( \Delta+i) }^{r-1}
            \\    &\ge  \sum_{i=1}^{m} \frac{1}{ \Delta+i } \cdot \frac{1}{ (r-1)! } \del{ \ln\del{ m + \Delta + 1 } - \ln( \Delta+i) }^{r-1}
            \\    &\ge  \int_{1}^{m+1} \frac{1}{ \Delta+x } \cdot \frac{1}{ (r-1)! } \del{ \ln\del{ m + \Delta + 1 } - \ln( \Delta+x) }^{r-1} \dif x
            \\    &=    \Bigl[~ - \frac{1}{r!} \del{ \ln(m+\Delta+1) - \ln(\Delta+x) }^r ~\Bigr]_{x=1}^{m+1}
            \\    &=    \frac{1}{r!} \del{ \ln(m+\Delta+1) - \ln(\Delta + 1) }^r
        \end{align*}
        which concludes the proof of the lemma.
        In order to deduce the theorem, we choose $\Delta = 0$ to obtain a
        random orientation (i.e., a probability distribution of
        orientations) of $G^0$ such that
        \begin{equation}
            \label{eq: bound TG0}
            \ex{ T(G^0) } \ge \frac{1}{r!} \ln^r( m+1 ) .
        \end{equation}
        A directed random walk on $G$ corresponds to a random walk on $G^0$,
        except that the latter does one additional step in the end (from the
        sink of $G$ to the extra vertex $v_\infty$).
        Thus we need to subtract $1$ from the bound \eqref{eq: bound TG0} to
        obtain the desired bound \eqref{eq: theorem USOs}.
        We are left only to observe that there must then also exist at
        least one concrete (not random) choice of orientation $G$ that
        satisfies this bound.
        This concludes the proof of Theorem~\ref{Theorem: USOs}.
\end{proof}

In Theorem~\ref{Theorem: USOs} we chose $n$ to be a multiple of $r$.
For other values of $n$ we can still deduce essentially the same bound, as in
the following corollary.

\begin{corollaryrep}
    For all $r,n \in \NN$ with $n > r$ there is an $r$-dimensional acyclic unique
    sink orientation of a grid $\hat G$ of size $n$ such that
    $
        \E\bigl[ T(\hat G) \bigr] > \frac{1}{r!} \ln^r \del{ \frac{n}{r} } - 1.
    $
\end{corollaryrep}

\begin{proof}
    If $n$ is divisible by $r$, then the corollary is an immediate consequence
    of Theorem~\ref{Theorem: USOs}.
    Assume that $n$ is not divisble by $r$, and let $m = \lfloor \frac{n}{r} \rfloor$.
    We take our construction $G$ from the proof of Theorem~\ref{Theorem: USOs} for the
    size $n' := rm$ and embed it into a grid $\hat G$ of size $n$.
    We can do this in such a fashion that all the edges between $G$ and $\hat
    G \setminus G$ point into $G$.
    In this way we obtain an acyclic unique sink orientation on $\hat G$ with
    $
        \E\bigl[ T(\hat G) \bigr]
        \ge \E\bigl[ T(G) \bigr]
    $,
    which yields the corollary using \eqref{eq: theorem USOs} with $m >
    \frac{n}{r}-1$.
\end{proof}

\section{One line and \texorpdfstring{$n$}{n} points}
\label{Section: one line and n points}

Here we describe the geometric setting in which we prove our main theorem.
We will assume that we are given a set of $n$ points $A \subseteq \RR^r$ and a
non-zero vector $u \in \RR^r$.
Its linear span $\RR u$ is a line: the ``one line'' or \emph{requirement line}
featured in the heading of this section.

\subparagraph{Pierced simplices and general position.}
We call a set $S \subseteq A$ \emph{pierced} or, more exactly,
\emph{$u$-pierced} if the convex hull $\conv(S)$ intersects the requirement
line $\RR u$.
If in addition $|S| = r$, then $S$ is a \emph{pierced simplex}.
Some readers might find it more natural to reserve the term ``simplex'' for
the set $\conv(S)$ instead of $S$; we will always take care to distinguish
between the two whenever the distinction is important.
A pierced simplex $S$ is \emph{non-degenerate} if
\textsc{(i)}
$S$ is affinely independent,
\textsc{(ii)}
no proper subset of $S$ is a pierced set, and
\textsc{(iii)}
the affine hyperplane spanned by $S$ and the
requirement line are not parallel.\sentenceDash
Within this paper, when we say that $(A,u)$ is \emph{in general position}, we
merely mean that every $u$-pierced simplex $S \subseteq A$ is non-degenerate.

\subparagraph{Below and above.}
Consider the affine span of a non-degenerate pierced simplex $S \subseteq A$:
it is a hyperplane in $\RR^r$ that is not parallel to the vector $u$.
The direction of $u$ thus determines an orientation (a positive and negative
side) of this hyperplane.
For any point $x \in A$, we will say that $x$ lies \emph{(strictly) below
$S$} if it lies on the (strictly) negative side of the hyperplane.
The word ``above'' is understood similarly.

\begin{figure}[t]
    \begin{minipage}[t]{0.55\textwidth}
        \vspace{0pt} 
        \centering
        \includegraphics[height=10em]{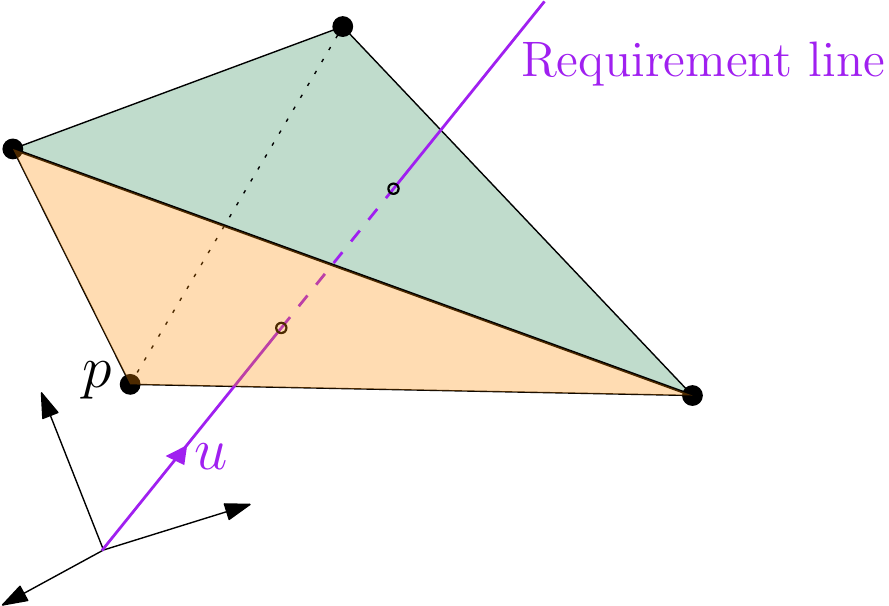}
    \end{minipage}
    \begin{minipage}[t]{0.44\textwidth}
        \vspace{0pt} 
        \caption{A tetrahedron in $\RR^3$.
            Its two front facets, the green and orange triangles,
            are both a pierced simplex.
            If $S$ denotes the green triangle at the top, then the \emph{simplex
            obtained by pivoting at $S$ with $p$} is the orange triangle at the
            bottom.}
        \label{Figure: pierced}
    \end{minipage}
\end{figure}

\subparagraph{Pivoting steps.}
Given $(A,u)$ in general position, a pierced simplex $S \subseteq A$ and
a point $p \in A$ strictly below $S$, we define a new pierced simplex $S'
\subseteq A$ which we call the \emph{simplex obtained by pivoting at $S$ with
$p$}.
To this end consider the set $S \cup \{p\}$:
It is an $r$-dimensional simplex, and the boundary of $\conv(S \cup \{p\})$
is pierced by the line $\RR u$ exactly twice: once in the facet $S$, and once
in another facet which we take to be $S'$.
See Figure~\ref{Figure: pierced} for an example.
We note that
\begin{itemize}
    \item $S'$ is by general position uniquely determined,
    \item $S'$ is a subset of $S \cup \{ p \} \subseteq A$, and
    \item $S'$ is a pierced simplex.
\end{itemize}

\subparagraph{The random process.}
Given a finite set $A \subseteq \RR^r$ of $n$ points, a non-zero vector $u \in
\RR^r$, and a $u$-pierced simplex $S_0 \subseteq A$,
we define the following random process, denoted $\Rr(A,u,S_0)$.
We will keep on assuming that $(A,u)$ is in general position.

The states (or positions) of the process are pierced simplices, and $S_0$ is
the starting position.
The consecutive positions $S_1,S_2,\dots$ are obtained as follows.
If the current position is $S_i$, let $p_i$ be a point (the $i$th ``pivot'')
chosen uniformly at random from the set of points from $A$ that lie strictly
below $S_i$.
(If there are no such points, then the random process terminates at this stage.)
Now define $S_{i+1}$ to be the simplex obtained by pivoting at $S_i$ with
$p_i$.\sentenceDash Our main theorem in Section~\ref{Section: analysis} states that
the expected number of steps until the process terminates can be of order
$\log^r(n)$.

\section{Construction}
\label{Section: construction}

Here we construct the point set that underlies the proof of
our main theorem (Theorem~\ref{Theorem: geometric lower bound} in
Section~\ref{Section: analysis}).
\ifArxivVersion
    For the proofs of the technical lemmas in this section we refer the reader to
    the appendix.
\else
    For the proofs of the technical lemmas in this section we refer the reader to
    the full version of this paper.
\fi

\subparagraph{Points, colors, layers, and phases.}
For all $r,m \in \NN$ we define our point set $A(r,m) \subseteq
\RR^r$ as follows;
a sketch is shown in Figure~\ref{Figure: pointset}.
We use the notation $\zero_r = (0,\dots,0)$ for the all-zeros vector in $\RR^r$.
We let
\[
    A(r,m) \defeq \{\, a_{i,j,k} \,:\, i,j \in [r],~ i \le j,~ k \in [m] \,\}
\]
where
\[
    a_{i,j,k} \defeq
        \pmat{
            \zero_{i-1}
            \\
            \del{ m^3 + m^5 + \dots + m^{2r-2j+1} } + (r-j)m + k
            \\
            \zero_{j-i}
            \\
            -m^{2r-2j+1}
            \\
            -m^{2r-2j-1}
            \\
            \vdots
            \\
            -m^5
            \\
            -m^3
        }.
\]
In particular, for $j=r$, we have $a_{i,r,k} = k \ee_i$, where $\ee_i$ denotes
the $i$th standard unit vector in $\RR^r$.
We call the indices $i,j,k$ the \emph{color}, the \emph{layer}, and the
\emph{phase} of a point, respectively.
Sometimes we will need a notational shorthand for colors and layers, so we define
$\Cc_i \subseteq A(r,m)$
to denote the set of points of color $i$,
and $\Ll_j \subseteq A(r,m)$
to denote the set of points from layer $j$.
So defined, $\Ll_j$ consists of $jm$ points, and our point set consists of
$n = \binom{r+1}{2} \cdot m$ points overall.

We will fix $u$ to denote the all-ones vector, $u = \one_r$, so that a set $S
\subseteq A(r,m)$ is pierced if and only if its convex hull intersects the
diagonal line $\RR \one_r$.
The rest of this section is devoted to a number of lemmas concerning the
relevant structure of our construction.
We begin by identifying the pierced subsets.
\begin{toappendix}
While the lemmas in this section could also be proved in an intuitive
geometric/visual way, such an exposition would be likely to take longer than
our straight-forward proofs based on calculations.
\end{toappendix}

\begin{lemmarep}
    \label{Lemma: colors}
    Let $S \subseteq A(r,m)$ be a pierced subset.
    Then $S$ contains a point of color $i$, for all $i \in [r]$.
\end{lemmarep}

\begin{appendixproof}
    This follows essentially from the fact that every point from our set $A(r,m)$
    has a unique positive (and large) entry that determines its color class.
    The formal proof goes as follows.
    Since $S$ is pierced, there is a convex combination of the form
    \begin{align}
        \label{eq: lambda-x}
        \mu \one_r = \sum_{x \in S} \lambda_x x
    \end{align}
    with $\mu \in \RR$, $\lambda_x \ge 0$ for all $x \in S$, and
    $\sum_{x \in S} \lambda_x = 1$.
    We will use the following sign properties that hold for every point $x \in
    A(r,m)$:
    \begin{romanenumerate}
    \item
        We have $\one_r^\T x > 0$.
    \item
        Let $i$ denote the color of the point $x$.
        Then $x_i > 0$, and $x_j \le 0$ for all $j \neq i$.
    \end{romanenumerate}
    Taking the sum of the coordinates in \eqref{eq: lambda-x} and using
    property (i), we obtain
    \begin{align*}
        \mu \hspace{1pt} r = \sum_{x \in S} \lambda_x \one_r^\T x = \one_r^T x > 0
    \end{align*}
    and hence $\mu > 0$.
    Comparing the signs of the left-hand and right-hand side in \eqref{eq:
    lambda-x} we find that for every $i$ there must be some $x^{(i)} \in S$
    with $x^{(i)}_i > 0$.
    By property (ii), $x^{(i)}$ has color $i$, which proves the lemma.
\end{appendixproof}

\begin{figure}
    \begin{minipage}[t]{0.70\textwidth}
        \vspace{2pt}
        \centering
        \includegraphics[height=12em]{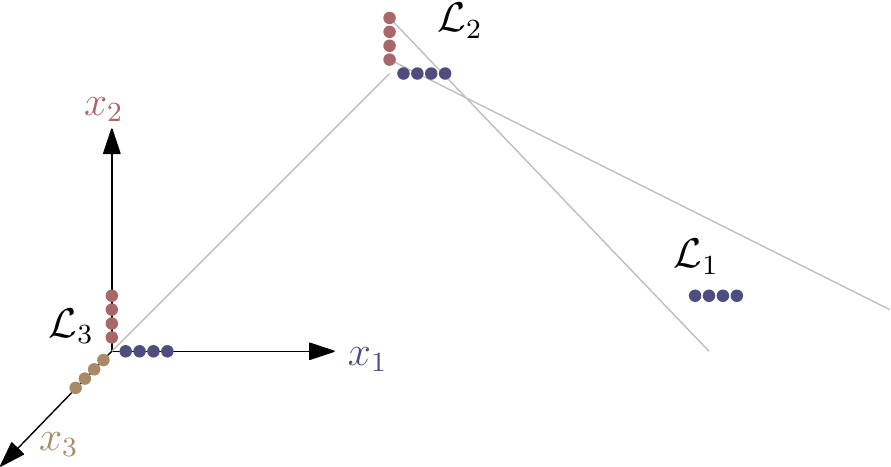}
    \end{minipage}
    \begin{minipage}[t]{0.29\textwidth}
    \caption{The constructed point set for $r=3$, $m=4$. The sketch is not
        true to scale. All off-axis points lie in the plane $x_3 = -64$.}
    \label{Figure: pointset}
    \end{minipage}
\end{figure}

The next lemma shows the converse of Lemma~\ref{Lemma: colors}, implying that the
set of pierced simplices can be identified with the set $\Cc_1 \times \dots
\times \Cc_r$.
More to the point, this implies that the extended Gale transform of our point
set defines a grid polytope.

\begin{lemmarep}
    \label{Lemma: pierced}
    Let $S \subseteq A(r,m)$ and assume that $S$ contains a point from each
    color class $\Cc_1,\dots,\Cc_r$.
    Then $S$ is a pierced subset.
    Furthermore, for each $i \in [r]$, $\conv(S)$ intersects the $i$th coordinate
    axis in some point $t \ee_i$ with $t > 0$.
\end{lemmarep}

\begin{appendixproof}
    It suffices to show the second statement of the lemma; so let $i \in [r]$.
    Consider the system
    \begin{align}
        \label{eq: matrix}
        B \lambda = t \ee_i
    \end{align}
    in the variables $(\lambda_1,\dots,\lambda_r,t) \in \RR^{r+1}$, where
    $B \in \RR^{r \times r}$
    is a matrix whose $j$th column is an arbitrarily chosen element of $S \cap
    \Cc_j$ ($j = 1,\dots,r$).
    We need to show that \eqref{eq: matrix} has a solution with $t > 0$,
    $\sum_i \lambda_i = 1$ and $\lambda_i \ge 0$ for all $i$.
    To this end we note that, by construction of our point set,
    \begin{itemize}
    \item
        all non-diagonal entries of $B$ are non-positive, and
    \item
        $B^\T$ is strictly diagonally dominant, i.e.,
        $B_{ii} > \sum_{j \neq i} \abs{ B_{ji} }$.
    \end{itemize}
    These conditions imply that $B^\T$ and hence $B$ is what is known as a
    non-singular M-matrix, which in turn implies that $B^{-1}$ exists and is
    non-negative (cf.~conditions $\text{N}_{39}$ and $\text{F}_{15}$ in
    \cite{Plemmons'77}).
    Hence the $i$th column of $B^{-1}$ is a non-negative solution
    $\hat{\lambda}$ to the equation $B \hat{\lambda} = \ee_i$.
    We can scale the vector $\hat{\lambda}$ by the factor
    \[
        \textstyle
        t \defeq \del{ \sum_i \hat{\lambda}_i }^{-1} > 0
    \]
    so as to achieve that the entries of the vector $\lambda \defeq t
    \hat{\lambda}$ sum up to $1$, which gives the desired solution to the
    system \eqref{eq: matrix}.
\end{appendixproof}

\begin{lemmarep}
    \label{Lemma: non-degenerate}
    $(A,\one_r)$ is in general position; that is,
    every pierced simplex $S \subseteq A(r,m)$ is non-degenerate.
\end{lemmarep}

\begin{appendixproof}
    First we need to prove that $S$ is affinely independent.
    To this end, let $\aff(S)$ denote the affine span of $S$:
    we want to prove that $\aff(S)$ is $(r-1)$-dimensional (a hyperplane).
    Assume that $\aff(S)$ was at most $(r-2)$-dimensional, and let $z$
    denote the intersection of $\conv(S) \subseteq \aff(S)$ with the
    requirement line.
    Then, by Carath\'eodory's theorem applied within the space $\aff(S)$,
    some subset $S' \subseteq S$ with $|S'| \leq r-1$ would still contain $z$
    in its convex hull.
    It follows that $S'$ is a pierced subset with at most $r-1$ elements; a
    contradiction to Lemma~\ref{Lemma: colors}.

    Second we need to check that no proper subset of $S$ is a pierced subset.
    Assuming the contrary would again lead to an immediate contradiction with
    Lemma~\ref{Lemma: colors}.

    Third we need to prove that $\aff(S)$ and the requirement line are not
    parallel.
    Assume they were parallel.
    Then, since they intersect, the requirement line would lie within the
    hyperplane $\aff(S)$, wherein it intersects the convex set $\conv(S)$.
    But if it intersects $\conv(S)$, then it also intersects some facet
    $\conv(S')$ with $S' \subseteq S$ and $|S'| \le \dim(\aff(S)) = r-1$.
    Again we would have that $S'$ is a pierced subset with at most $r-1$
    elements, in contradiction with Lemma~\ref{Lemma: colors}.
\end{appendixproof}

Lemma~\ref{Lemma: non-degenerate} above assures that the random process
associated with our construction is well-defined.
The next lemma states that, as our random process evolves, the intersection
value $t$ of the current position with the $i$th coordinate axis ($i =
1,\dots,r)$ is monotonically decreasing with time and, thus, can serve as a
measure of progress.
This is of course by no means true for an arbitrary point set, but it holds in
the case of our construction.
Next, Lemma~\ref{Lemma: layer r-1} states a simple condition from which to
tell whether the points from the layer $\Ll_{r-1}$ lie above or below the
current position; this condition is immediately relevant for
the analysis of the random process.

\begin{lemma}
    \label{Lemma: monotone}
    Let $i \in [r]$.
    Let $S \subseteq A(r,m)$ be a pierced simplex,
    let $p \in A$ be a point strictly below $S$, and let $S'$ denote
    the pierced simplex obtained by pivoting at $S$ with $p$.
    Let $t_i \ee_i$ and $t'_i \ee_i$ denote the intersection of $\conv(S)$
    (respectively, $\conv(S')$) with the $i$th coordinate axis, as in
    Lemma~\ref{Lemma: pierced}.
    Then we have $t'_i \le t_i$ for all $i$.
\end{lemma}

\begin{toappendix}
We omit the proof of Lemma~\ref{Lemma: monotone}; the reader may want to make it
visually plausible with the help of Lemma~\ref{Lemma: pierced}.
The latter is essentially saying that all points of
$A(r,m)$ lie on the outside of the cone spanned by the coordinate axes.
\end{toappendix}

\begin{lemmarep}
    \label{Lemma: layer r-1}
    Let $r \ge 2$, $m \ge 2$, let $S \subseteq A(r,m)$ be a pierced simplex,
    and let $t_i \ee_i$ denote the intersections  of $\conv(S)$ with the $i$th
    coordinate axis as in Lemma~\ref{Lemma: pierced} ($i = 1,\dots,r$).
    \addtolength\leftmargini{0.5em}
    \begin{alphaenumerate}
    \item
        If, for some $i$, $t_i \le t_r$,
        then all points from $\Cc_i \cap \Ll_{r-1}$ lie strictly \emph{above} $S$.
    \item
        If, for some $i$, $t_i \ge t_r+1$,
        then all points from $\Cc_i \cap \Ll_{r-1}$ lie strictly \emph{below} $S$.
    \item
        If $i$ is such that $t_i = \min \{ t_1,\dots,t_r \}$, then all points
        from $\Cc_i \setminus \Ll_r$ lie strictly \emph{above} $S$.
        In particular we then have $t_i \ee_i \in S$.
    \end{alphaenumerate}
\end{lemmarep}

\begin{appendixproof}
    First we recall that $S$ contains one point from each color class, and
    we remark that the color class $\Cc_r$ consists solely of the points
    $\ee_r, 2\ee_r, \dots, m\ee_r$.
    This implies $1 \le t_r \le m$, which will be used in the calculation
    below.
    Second we recall that a generic point from $\Cc_i \cap \Ll_{r-1}$ is, by
    definition, of the form
    \begin{align}
        \label{eq: generic 1}
        \pmat{
            \zero_{i-1}
            \\
            m^3 + m + k
            \\
            \zero_{r-1-i}
            \\
            -m^3
        }
    \end{align}
    with $k \in [m]$.

    \emph{Ad (a).}
    The affine hyperplane spanned by $S$ has the equation
    \begin{align}
        \label{eq: hyperplane}
        \frac{x_1}{t_1} + \dots + \frac{x_r}{t_r} - 1 = 0;
    \end{align}
    and a point lies strictly above $S$ if and only if the left hand side in
    \eqref{eq: hyperplane} is positive.
    We can easily verify the statement by plugging \eqref{eq: generic 1}
    into \eqref{eq: hyperplane}, which gives indeed
    \begin{align*}
        \frac{ m^3 + m + k }{ t_i }
        +
        \frac{ -m^3 }{ t_r }
        - 1
        &\ge
        \frac{ m^3 + m + 1 }{ t_r }
        +
        \frac{ -m^3 }{ t_r } - 1
        \\&=
        \frac{  m + 1 }{ t_r } - 1
        \ge
        \frac{  m + 1 }{ m } - 1
        > 0.
    \end{align*}

    \emph{Ad (b).}
    This time we need to show that the left-hand side in \eqref{eq: hyperplane}
    is negative. We calculate similarly:
    \begin{align*}
        \frac{ m^3 + m + k }{ t_i }
        +
        \frac{ -m^3 }{ t_r }
        - 1
        &\le
        \frac{ m^3 + 2m }{ 1 + t_r }
        +
        \frac{ -m^3 }{ t_r }
        - 1
        \\&=
        \frac{-m^3 + 2m t_r }{ (1 + t_r) t_r } - 1
        \\&\le
        \frac{ -m^3 + 2 m^2 } { 2 } - 1 &\hspace{-5em} \text{since $1 \le t_r \le m$}
        \\&<
        0.
    \end{align*}

    \emph{Ad (c).}
    Let $x \in \Cc_i \setminus \Ll_r$.
    We use the convenient fact that, for all such points $x$, we have $\one^\T
    x \ge m+1$.
    Plugging $x$ into \eqref{eq: hyperplane} gives
    \begin{align*}
        \frac{x_1}{t_1} + \dots + \frac{x_r}{t_r} - 1
        &= \underbrace{ \frac{x_i}{t_i} }_{ > 0 } + \sum_{j \neq i}
        \underbrace{ \frac{x_j}{t_j} }_{ < 0 } - 1
        \ge \frac{ x_i + \sum_{j \neq i} x_j }{ t_i } - 1
        \\
        &= \frac{\one^\T x}{t_i} - 1
        \ge \frac{m+1}{m} - 1
        > 0.
    \end{align*}
    In order to deduce the additional statement $t_i \ee_i \in S$,
    recall that by Lemma~\ref{Lemma: colors}, $S$ must contain some point from
    the color class $\Cc_i$, and by the above, this point must be an element
    of $\Ll_r$. But all the points from $\Ll_r \cap \Cc_i$ lie on the $i$th
    axis, so $t_i \ee_i$ is the only possible choice.
\end{appendixproof}

The last lemma in this section states that taking the point set $A(r+1,m)$ and
removing the outermost layer $\Ll_{r+1}$ yields a point set that is
equivalent, for our purposes, to the set $A(r,m)$.
This observation is key to the inductive approach followed in
section~\ref{Section: analysis}.
Actually, the statement is a bit more general: The lemma starts from the set
$A(R,m)$ for any $R > r$ and then removes all higher layers
$\Ll_{r+1},\dots,\Ll_{R}$.

\begin{lemmarep}
    \label{Lemma: deep}
    Assume $m \ge 3$.
    For $R > r$,
    let
    \[
        B := \{ (x_1,\, \cdots,\, x_r) \,:\, x \in A(R,m),~
            x \in \Ll_{1} \cup \dots \cup \Ll_r \}.
    \]
    Then Lemmas~\ref{Lemma: colors} to \ref{Lemma: layer r-1} are also valid
    for the point set $B$ in place of $A(r,m)$.
\end{lemmarep}

\begin{appendixproof}
    Lemmas~\ref{Lemma: colors} to~\ref{Lemma: monotone} are based on sign
    properties that still hold for the set $B$; so their validity follows with
    the original proof.
    We will now verify statements (a,b,c) from Lemma~\ref{Lemma: layer r-1} for the
    set $B$ by showing that their original proof can be adapted,
    where we abbreviate
    \[
        \alpha := m^3 + m^5 + \dots + m^{2R-2r+1} + (R-r)m.
    \]
    The hyperplane spanned by $S$ can again be written as in \eqref{eq:
    hyperplane}, but this time instead of $1 \le t_r \le m$
    we have $\alpha+1 \le t_r \le \alpha+m$,
    and this time a generic point from $\Ll_{r-1} \cap \Cc_i$ is of the form
    \begin{align}
        \label{eq: generic 2}
            \pmat{
                \zero_{i-1}
                \\
                \alpha + m^{2R-2r+3} + m + k
                \\
                \zero_{r-1-i}
                \\
                - m^{2R-2r+3}
            }.
    \end{align}

    \emph{Ad (a).}
    We assume $t_i \le t_r$.
    In order to show that all points from $\Cc_i \cap \Ll_{r-1}$ lie below the
    hyperplane given by \eqref{eq: hyperplane},
    we plug \eqref{eq: generic 2} into \eqref{eq: hyperplane} to obtain
    \begin{align*}
        &
        \frac{ \alpha + m^{2R-2r+3} + m + k }{ t_i }
        +
        \frac{ -m^{2R-2r+3} }{ t_r }
        - 1
        \\&\ge
        \frac{ \alpha + m^{2R-2r+3} + m + 1 }{ t_r }
        +
        \frac{ -m^{2R-2r+3} }{ t_r }
        - 1
        \\&\ge
        \frac{ \alpha + m + 1 }{ \alpha+m }
        - 1
        =
        \frac{1}{\alpha+m}
        > 0,
    \end{align*}
    as desired.

    \emph{Ad (b).}
    We first note that the expression $\alpha+m+2$, which will occur in the
    calculation below, can be upper bounded e.g.~by $\frac{89}{72}
    m^{2R-2r+1}$ (assuming $m \ge 3$ and $R-r \ge 1$),
    which implies the coarse, but useful, bound
    \begin{align}
        \label{eq: useful bound}
        (\alpha+m+2)(\alpha+m+1) < 2 m^{4R-4r+2}.
    \end{align}
    This time we assume $t_i \ge 1 + t_r$ and we want to show that
    all points from $\Cc_i \cap \Ll_{r-1}$ lie below the hyperplane spanned by
    $S$; so we plug a generic such point into \eqref{eq: hyperplane} to obtain
    \begin{align*}
        &
        \frac{ \alpha + m^{2R-2r+3} + m + k }{ t_i }
        +
        \frac{ -m^{2R-2r+3} }{ t_r }
        - 1
        \\&\le
        \frac{ \alpha + m^{2R-2r+3} + 2m }{ 1 + t_r }
        -
        \frac{ m^{2R-2r+3} }{ t_r }
        - 1
        \\&=
        \frac{ \alpha + 2m }{ 1 + t_r }
        -
        \frac{ m^{2R-2r+3} }{ (1 + t_r) t_r }
        - 1
        \\&\le
        \frac{ \alpha + 2m }{ \alpha+2 }
        -
        \frac{ m^{2R-2r+3} }{ (\alpha+m+2) (\alpha+m+1) }
        - 1
        \\&=
        \frac{ 2m-2 }{ \alpha+2 }
        -
        \frac{ m^{2R-2r+3} }{ (\alpha+m+2) (\alpha+m+1) }
        \\&<
        \frac{ 2m-2 }{ m^{2R-2r+1} }
        -
        \frac{ m^{2R-2r+3} }{ 2 m^{4R-4r+2} }
        &\hspace{-5em} \text{by \eqref{eq: useful bound}}
        \\&=
        \frac{-1}{ 2 m^{2R-2r+1} } \del{ m-2 }^2 \le 0.
    \end{align*}

    \emph{Ad (c).}
    The proof from Lemma~\ref{Lemma: layer r-1} adapts immediately, this time
    using the relation $\one^\T x \ge \alpha+m+1$.
\end{appendixproof}

\section{Analysis}
\label{Section: analysis}

The goal of this section is to prove the main theorem of this paper, concerned
with the random process $\Rr_{r,m} = \Rr(A,u,S_0)$ associated with the point set
$A \defeq A(r,m)$, the all-ones vector $u \defeq \one_r$, and the starting position
$S_0 \defeq \{ m\ee_1,\dots,m\ee_r \}$:
\begin{theorem}
    \label{Theorem: geometric lower bound}
    The expected number of steps performed by the random process $\Rr_{r,m}$ is at least
    \begin{align*}
        \frac{1}{r!^3} \del{ \ln \del[1]{  m + \chooseTwo{r} + \Delta } - \ln \del[1]{ 1+ \chooseTwo{r} + \Delta } }^r
        = \Omega(\log^r m).
    \end{align*}
\end{theorem}
In terms of the number of points, $n$, the bound
can be written in the form $\Omega(\log^r n)$.

\subparagraph{The augmented process $\Rr^\Delta_{r,m}$.}
In order to make an inductive proof possible, we will
make use of an ``augmented'' pivoting process,
in analogy to the ``augmented graph'' that we used in the proof of
Theorem~\ref{Theorem: USOs}.
Given a number $\Delta \ge 0$, the \emph{augmented process} $\Rr^\Delta_{r,m}$
is defined as follows.
\begin{itemize}
    \item
        The starting position is chosen by an adversary, in the following way.
        The adversary chooses one new point $\alpha_i \ee_i$ on each axis,
        subject to the constraint $\alpha_i \ge m$.
        These points are added to the point set $A = A(r,m)$ to obtain an
        augmented point set
        \[
            A' \defeq A \cup \{ \alpha_i \ee_i \,:\, i \in [r] \},
        \]
        and the starting position is now chosen as
        $
            S_0 \defeq \{ \alpha_1 \ee_1, \dots, \alpha_r \ee_r \}.
        $
        We fit the new points into our terminology of colors, layers and
        phases by saying that $\alpha_i \ee_i$ has color $i$, layer $r$ and
        phase $m+1$; and we remark that the lemmas in Section~\ref{Section:
        one line and n points} still hold for the augmented point set.
    \item
        The positions $S_0,S_1,\dots$ of the augmented process are pierced
        simplices of $A'$, except that we also introduce a new, special,
        position $S_\infty$.
        ($S_\infty$ is just a formal symbol; it is not represented by any
        simplex.)
        This will be the terminal position.
    \item
        If we are currently at position $S_t$, then the next position $S_{t+1}$ is
        obtained as follows:
        Let $\text{below}(S_t) \subseteq A'$ denote the set of points that lie
        strictly below $S_t$.
        We draw a pivot element $p_{t+1}$ from the set $\text{below}(S_t) \cup
        \{ \infty \}$ according to the distribution
        \begin{align*}
            \pr{p_{t+1} = x} =
                \begin{cases}
                    \frac{1}{\abs{ \text{below}(S_t) } + \Delta}
                        & \text{ for } x \in \text{below}(S_t),
                    \\[2ex]
                    \frac{\Delta}{ \abs{ \text{below}(S_t) } + \Delta}
                        & \text{ for } x = \infty.
                \end{cases}
        \end{align*}
        If $p_{t+1} = \infty$, then $S_{t+1} = S_\infty$, and the process terminates.
        Otherwise we perform a standard pivoting step at $S_t$ with $p_{t+1}$.
        (Edge case: If $\Delta = 0$ and $\text{below}(S_t) = \emptyset$, then
        we always pick $p_{t+1} = \infty$.)
\end{itemize}
Note that, despite its name, the ``augmented'' process typically terminates
earlier than the non-augmented process: the larger the parameter $\Delta$ is,
the sooner!
For $\Delta=0$ the augmented process behaves like the original, non-augmented
process\midSentenceDash except for the modified starting position and one
additional final pivoting step towards the terminal position $S_\infty$.

\subparagraph{The phase of a pierced simplex.}
We define the \emph{phase} of a pierced simplex $S \subseteq A'$ as
\[
    \phase(S) \defeq \min \{ \phase(p) \,:\, p \in S \text{ with }
    \mathrm{layer}(p) = r \}.
\]
The minimum is well-defined because $S$ contains at least
one point from the layer $\Ll_r$:
Indeed Lemma~\ref{Lemma: colors} tells us that $S$ contains a point from
$\Cc_r$, which is a subset of $\Ll_r$.
For consistency we also define $\text{phase}(S_\infty) \defeq 0$.\sentenceDash
Note that we are overloading the term ``phase'', because we have defined
the phase of a point earlier.

We remark that the phase of a pierced simplex can be equivalently
written as $\phase(S) = \min \{t_1,\dots,t_r\}$, where $t_i \ee_i$ denotes the
intersection of $\conv(S)$ with the $i$th coordinate axis;
this follows from Lemma~\ref{Lemma: layer r-1}(c).
Using this observation, the possible choices for pivots at the time of a phase
change are easily identified, and from this information we may read off the
probability that a particular phase is visited.
The following lemma summarizes the result of this observation.

\begin{lemma}[--- The phases visited by the augmented process]
    \label{Lemma: phase process}
    Let $\sigma_1 < \sigma_2 < \dots < \sigma_N$ denote the times at which a phase change
    occurs in the augmented random process, and let $(\phi_i)_{0 \le i \le N}$
    denote the phases that are visited, i.e.,
    $
        \phi_0 = m+1,
    $ and
    $
        \phi_i = \phase(S_{\sigma_i})
    $
    ($1 \le i \le N$).
    Then we have, for $i \ge 1$:

    \begin{alphaenumerate}
        \item
            The distribution of $\phi_{i}$ is given by
            \begin{align*}
                \pr{ \phi_i = x \Given \phi_{i-1} } =
                \begin{cases}
                    \frac{ r }{ r(\phi_{i-1} - 1) + \Delta }
                    & \text{ for } x \in [\phi_{i-1} - 1],
                    \\[2ex]
                    \frac{ \Delta }{ r(\phi_{i-1} - 1) + \Delta }
                    & \text{ for } x = 0.
                \end{cases}
            \end{align*}
        \item
            If $\phi_i > 0$, then the color of the pivot at time
            $\sigma_i$ is a u.a.r.~element of $[r]$.
    \end{alphaenumerate}
\end{lemma}

Consider one of the phases $\phi_i$ that are visited by the augmented random
process.
We want to bound the \emph{duration} of the phase $\phi_i$, i.e.~the
number $\sigma_{i+1} - \sigma_i$, from below.
In general this duration could be very short, so we introduce a suitable
notion of a \emph{good phase}.
The definition will guarantee that, when entering a good phase, all points of
the layer $\Ll_{r-1}$ will lie strictly below the current position; and this
property will in turn make it possible to derive a lower bound on the duration
of a good phase inductively.

\begin{definition}[(good phases)]
     Let $k \in [m]$.
     We say that $k$ is a \emph{good phase} of the augmented process if,
     using the notation from Lemma~\ref{Lemma: phase process},
     \begin{romanenumerate}
         \item
             $k$ is visited, so that $k = \phi_j$ for some $j$,
         \item
             the pivot at time $\sigma_j$ has color $r$, and
         \item
             the position $S_{\sigma_j}$ does not contain any point from the
             layer $\Ll_{r-1}$.
     \end{romanenumerate}
\end{definition}

Let the reader be warned that the above definition is weaker than one might
think at first: The only points that we take directly into account are those
from the two outermost layers $\Ll_r$ and $\Ll_{r-1}$.
In particular we allow $S_{\sigma_j}$ to contain points from other layers.
When reading on, it is useful to keep in mind one consequence of
Lemma~\ref{Lemma: layer r-1}(c):
The phase of the current position can change only when pivoting a point from
the layer $\Ll_r$, i.e., a point that lies on one of the coordinate axes.
Consequently, pivots in lower layers can largely be ignored in our analysis.

\newcommand{\DeltaByM}{{\textstyle\frac{\Delta}{m}}}

\begin{lemma}
    \label{Lemma: good phase}
    Let $1 \le k \le m-1$.
    Then phase $k$ is a good phase with probability at
    least
    $
        \frac{ 1 }{ r (\Delta + kr) }.
    $
\end{lemma}

\begin{proof}
    Let $j$ be the (random) largest index such that $\phi_{j-1} > k$.
    Then the probability of (i) equals $\pr{ \phi_j = k }$, and
    using Lemma~\ref{Lemma: phase process}(a) we compute this probability to be
    $\frac{ r }{ rk + \Delta }$.
    Given (i), Lemma~\ref{Lemma: phase process}(b) tells us that the
    probability of (ii) equals $1/r$.

    Assume that (i) and (ii) hold.
    It remains to show that, in this case, (iii) holds with
    probability at least $1/r$.
    We consider the position $S_{\sigma_j - 1}$ one time step before entering
    phase $k$.
    Using the same notation as in Lemma~\ref{Lemma: layer r-1}, let $t_i \ee_i$
    denote the intersections of $\mathrm{conv}(S_{\sigma_j - 1})$ with the
    $i$th coordinate axis ($i = 1,\dots,r$).
    Note that $t_i > k$ for all $i$, because phase $k$ has not been entered
    yet at this time.

    By Lemma~\ref{Lemma: layer r-1} it is sufficient to give a bound for the event
    \begin{align}
        \label{eq: axis intersection property}
        t_i \le t_r \text{ for all } i = 1,\dots,r-1.
    \end{align}
    To this end, let $\tau$ denote the time that the point $t_r \ee_r$ is
    pivoted, so that $S_\tau$ is the first position to include the point $t_r
    \ee_r$.
    Note that $t_r$ does not change in between time $\tau$ and time
    $\sigma_j$; thus, if property \eqref{eq: axis intersection property}
    already holds at time $\tau$, then by monotonicity (Lemma~\ref{Lemma:
    monotone}) it will still hold at time $S_{\sigma_j - 1}$.
    So assume that at time $\tau$ property \eqref{eq: axis intersection
    property} does not yet hold, so that there are some ``bad'' indices $i$
    with $t_i > t_r$.
    Let $I \subseteq [r-1]$ denote the set of such bad indices,
    and let $\tau_1 > \tau$ be the first time that another point of layer $r$
    with phase $\le t_r$ and color contained in $I \cup \{r\}$ is pivoted.
    With probability at least $\frac{ |I| }{ |I|+1 }$, the color of this pivot
    is contained in $I$, in which case the number of bad indices is reduced by
    $1$.
    Iterating this argument, the number of bad indices will be reduced down to
    zero with probability at least
    \begin{align*}
        \frac{ |I| }{ |I|+1 } \cdot
        \frac{ |I|-1 }{ |I| } \cdots
        \frac{ 1 }{ 2}
        = \frac{ 1 }{ |I|+1 }
        \ge \frac{ 1 }{ r },
    \end{align*}
    as desired.
\end{proof}

When the augmented process enters a good phase $k$, then all the points of the
layer $\Ll_{r-1}$ lie strictly below the current position.
We restrict our attention to the hyperplane $x_r = -m^3$ that contains all the
points from $A(r,m) \setminus \Ll_r$:
Due to Lemma~\ref{Lemma: deep}, the augmented process within this hyperplane
behaves like the augmented process on the lower-dimensional construction
$A(r-1,m)$.
However, the process might at any point pivot one of the points $\{ p \in
\Ll_r \,:\, \text{phase}(p) < k \}$, and as soon as this happens, the good
phase $k$ already ends.
We can account for this by adjusting the parameter $\Delta$ and we obtain the
following lemma.

\begin{lemma}
    \label{Lemma: coupling}
    For every $1 \le k \le m-1$, if phase $k$ is visited and if it is a good
    phase, then its expected duration is bounded from below by the
    best-case\footnote{The term ``best-case'' here, as well as in
        Theorem~\ref{Theorem: augmented process bound}, refers to the action of the
        adversary who chooses the starting position of the augmented process.
        The intended meaning is for the lower bound in Theorem~\ref{Theorem:
        augmented process bound} to hold for any choice of starting position.}
    expected duration of the process $\Rr^{\Delta + (k-1)r}_{r-1,m}$.
\end{lemma}

\begin{theorem}
    \label{Theorem: augmented process bound}
    Let $t^\Delta_{r,m}$ denote the best-case expected duration of the
    augmented process $\Rr^{\Delta}_{r,m}$. Then we have
    \begin{align*}
        t^\Delta_{r,m} \ge
        \frac{ 1 }{ r!^3 } \cdot
        \del{ \ln(m+ \chooseTwo{r} + \Delta) - \ln(1 + \chooseTwo{r} +\Delta) }^r.
    \end{align*}
\end{theorem}

\begin{proof}
    By induction on $r$.
    For $r=1$ the statement is easy to verify; let now $r \ge 2$.
    Combining Lemmas~\ref{Lemma: good phase} and~\ref{Lemma: coupling} we
    obtain
    \begin{align}
        \label{eq: trm}
        t^\Delta_{r,m}
        &\ge
        \sum_{k=1}^{m-1}
        \frac{ 1 }{ r (kr + \Delta) }
        \cdot t^{\Delta + kr - 1 }_{r-1,m}.
    \end{align}
    The induction hypothesis gives, for $1 \le k \le m-1$,
    \begin{align}
        t^{\Delta + kr - 1 }_{r-1,m}
        &\ge
        \frac{ 1 }{ (r-1)!^3 }
        \cdot
        \del{ \ln(m + \chooseTwo{r-1} + \Delta+kr-1) - \ln(\chooseTwo{r-1} + \Delta+kr) }^{r-1}.
        \nonumber\\
        &=
        \frac{ 1 }{ (r-1)!^3 }
        \cdot
        \del[2]{ \ln(m + \chooseTwo{r} + \Delta+kr-r) - \ln( \underbrace { 1 + \chooseTwo{r} + \Delta+kr-r }_{=: f(k)} ) }^{r-1},
        \label{eq: indhyp}
    \end{align}
    where we have used Pascal's rule to handle the binomial coefficients.
    Plugging \eqref{eq: indhyp} into \eqref{eq: trm} and furthermore using the
    simple inequality $kr + \Delta \le f(k)$, we thus obtain
    \begin{align}
        t_{r,m}^\Delta
        &\ge
            \frac{ 1 }{ (r-1)!^3 }
            \sum_{k=1}^{m-1}
            \frac{ 1 }{ r f(k) }
            t^{\Delta + kr - 1 }_{r-1,m}
            \nonumber
        \\
        &\ge
            \frac{ 1 }{ (r-1)!^3 }
            \int_{x=1}^{m}
            \frac
                { 1 }
                { r f(x) }
            { \del{ \ln(m + \chooseTwo{r} + \Delta + rx-r) - \ln f(x) }^{r-1} }
            \dif{x}
            \label{eq: before drop}
        \\
        &\ge
            \frac{ 1 }{ (r-1)!^3 }
            \int_{x=1}^{ 1 + (m-1)/r }
            \frac
                { 1 }
                { r f(x) }
            { \del{ \ln(m + \chooseTwo{r} + \Delta) - \ln f(x) }^{r-1} }
            \dif{x}
            \label{eq: after drop}
        \\
        &=
            \frac{ 1 }{ r!^3 } \cdot
            \Bigl[\, -
                \del{ \ln(m + \chooseTwo{r} + \Delta) - \ln f(x) }^r
            ~\Bigr]_{x=1}^{1 + (m-1)/r }
            \nonumber
        \\
        &=
            \frac{ 1 }{ r!^3 } \cdot \del{ \ln(m+ \chooseTwo{r} + \Delta) - \ln(1 + \chooseTwo{r} +\Delta) }^r,
            \nonumber
    \end{align}
    which proves the theorem.
    Note that the integrand in \eqref{eq: before drop} is positive everywhere,
    so we were justified to restrict the range of the integral in \eqref{eq:
    after drop}, effectively dropping negligible terms.

    Theorem~\ref{Theorem: geometric lower bound} now follows from
    Theorem~\ref{Theorem: augmented process bound} by setting $\Delta = 0$.
\end{proof}

\section{Conclusion}

\subparagraph{Outlook.}
It remains an open question whether one can obtain good upper bounds for
the expected number of steps performed by \textsc{Random-Edge} when the corank
is bounded.
The only non-trivial result at this point remains the $O(\log^2 n)$ bound
for the case $r = 2$, which was settled in
\cite{GaertnerSTVW'03} and which the author has studied further in a more
abstract setting in \cite{Milatz'17}.
It might well be  that there is a threshold behavior in the sense that
\textsc{Random-Edge} performs well for slowly growing $r$, and badly for
quickly growing $r$.
Finally, the same questions can be asked for other simplex pivoting rules
that might be easier to analyze.

\subparagraph{Remark on the dependence on $r$.}
The leading factor $1 / r!^3 $ in Theorem~\ref{Theorem: augmented process
bound} is rather small.
Due to the results by Friedmann et al.~\cite{FriedmannHZ'11}, this factor
cannot in general be tight.
Unfortunately, an improvement seems to be beyond the scope of our method.
For the interpretation of Theorem~\ref{Theorem: polytopes} some readers may
find it interesting to pick a value of $r$ that depends on the number of
facets $n$.
Not every such choice of $r$ leads to a meaningful bound;
but it is possible to choose
$ r = r(n) = \ln^{1/s} \del{ 1 + 4n/r^4 } $
with any $s > 3$, which leads to a lower bound of the form
$
    (\ln n)^{ \displaystyle \Omega( \ln^{1/s} n ) }
$
for the directed random walk on a grid polytope with $n$ facets and corank $r(n)$.

\subparagraph{Acknowledgements.}
I would like to thank Bernd Gärtner, Ahad N.~Zehmakan, Jerri Nummenpalo and
Alexander Pilz for many useful suggestions and discussions.

\bibliography{p60-Milatz}


\end{document}